\documentclass[a4paper,11pt]{amsart}

\usepackage{amsmath,amssymb,amsfonts,latexsym}
\usepackage[english]{babel}
\usepackage[dvipsnames]{xcolor}
\usepackage{multirow}
\usepackage[left=2.5cm,right=2.5cm,top=4cm,bottom=2.5cm]{geometry}
\usepackage{tikz}
\usepackage{pgf,tikz}
\usepackage{pgfplots}
\usepackage{pstricks-add}
\usepackage{enumitem}
\usepackage{textcomp}
\usepackage{colortbl}
\usepackage{slashbox}
\usepackage{setspace}
\usetikzlibrary{pgfplots.statistics}

\usepackage{arydshln}
\usetikzlibrary{matrix,shapes,arrows,positioning}
\definecolor{myblue}{RGB}{209,221,243}
\definecolor{myblu}{RGB}{110,148,216}
\definecolor{myre}{RGB}{255,188,188}
\definecolor{li}{gray}{0.97}
\definecolor{lt}{gray}{0.7}
\usepackage{tabu}
\usepackage{pifont}
\usepackage{hhline}
\usepackage{algorithm}
\usepackage[noend]{algpseudocode}
\makeatletter
\def\BState{\State\hskip-\ALG@thistlm}
\usetikzlibrary{shadows}
\usepackage{booktabs}
\newtheorem{theorem}{Theorem}[section]

\renewcommand{\algorithmicrequire}{\textbf{Input: }}
\renewcommand{\algorithmicensure}{\textbf{Output: }}
\usetikzlibrary{patterns}
\algrenewcommand\algorithmicindent{1.4em}%
\begin{document}

\title[An ideal hierarchical secret sharing scheme]{\textbf{\vspace{-3mm} \\}An ideal hierarchical secret sharing scheme}

\author[GHANEM Meriem  \and BOUROUBI Sadek]{GHANEM Meriem$^1$, BOUROUBI Sadek$^2$ \vspace{2mm} \\ \MakeLowercase{\texttt{ghanem.meriem@gmail.com}$^1$, \texttt{sbouroubi@usthb.dz}$^2$}  \vspace{3mm}\\\emph{$^{1,2\,}$DGRSDT, USTHB, F\MakeLowercase{aculty of }M\MakeLowercase{athematics}, D\MakeLowercase{epartment of} O\MakeLowercase{perations }R\MakeLowercase{esearch},\vspace{1mm} \\ L'IFORCE L\MakeLowercase{aboratory.}\vspace{1mm} \\ P.B. 32 E\MakeLowercase{l-Alia}, 16111, B\MakeLowercase{ab }E\MakeLowercase{zzouar}, A\MakeLowercase{lgiers}, A\MakeLowercase{lgeria.\vspace{1mm}\\
}}
}

\begin{abstract}
 One of the methods used in order to protect a secret $K$ is a \emph{secret sharing scheme}. In this scheme the secret $K$ is distributed among a finite set of participants $P$ by a special participant called the \emph{dealer}, in such a way that only predefined subsets of participants can recover the secret after collaborating with their secret shares. The construction of secret sharing schemes has received a considerable attention of many searchers whose main goal was to improve the information rate. In this paper, we propose a novel construction of a secret sharing scheme which is based on the hierarchical concept of companies illustrated through its organization chart and represented by a tree. We proof that the proposed scheme is \emph{ideal} by showing that the information rate equals 1. In order to show the efficiency of the proposed scheme, we discuss all possible kinds of attacks and proof that the security in ensured. Finally, we include a detailed didactic example for a small company organization chart. \\

\vspace{1mm}

\vspace{1em}

\end{abstract}



\maketitle

\section{Introduction}
\label{sec:1}
The fast development of computer networks and data communication systems make the protection of secret data extremely imperative. In order to protect a secret, several methods have been applied before, one of theme is to encrypt data, but this will change the problem instead of solving it, since another method is required to protect the encrypted data. It’s also possible to keep the secret in one well-guarded location, but this method is very unreliable since the secret can be destroyed or become inaccessible. Another method consists in sharing the data, either by storing multiple copies of the data in different locations, which would increase security vulnerabilities, or by splitting the data into several parts and sharing them between different members of the system. This last method is called secret sharing scheme and would be very efficient in case where the reconstruction of the initial data does not require the presence of all the system members, otherwise the veto given to each member would paralyze the system \cite{01}. Secret sharing schemes have many applications in different areas, such as access control, launching a missile, and opening a bank vault. For more details see for instance \cite{16,15}.\vspace{0.2cm} \\
The secret sharing scheme is therefore a method of distributing a secret $K$ among a finite set of participants $P$, in such a way that only predefined subsets of participants can collaborate with their secret shares to recover the secret $K$. These subsets are called \emph{qualified subsets} and the set of all qualified subsets is called the \emph{access structure} denoted $\Gamma$ \cite{07}. Each subset of participants $Y\in \Gamma$ is called \emph{a minimal qualified subset} if ($Y'\subset Y$ and $Y'\in \Gamma$) implies $Y'=Y$. The family of all minimal qualified subsets is noted $\Gamma_{0}$.
In a secret sharing scheme, the secret $K$ is chosen by a special participant, called the dealer, who is responsible for computing and distributing the shares among the set of participants $P$. The share of any participant refers specifically to the information that the dealer sends in private. It is required to keep the size of shares as small as possible since the security of a system degrades as the amount of information that must be kept secret increases.\vspace{0.2cm} \\
Many approaches have been proposed for the construction of a secret sharing scheme \cite{17}. The first one called $(t, n)$-\emph{threshold scheme} was introduced independently by Shamir and Blakley \cite{01,05} in 1979. In a $(t, n)$-threshold scheme, all groups of at least $t$ participants of $n$-participants are qualified and can reconstruct the secret, while those with less than $t$ participants are unqualified and can't have any information about the secret. The scheme proposed by Shamir is based on polynomials over a finite field $GF(q)$ since a random polynomial $f$ is chosen by the dealer for computing and distributing the shares among the set of participants $P$ in such a way that, each participant $p_{i}$ is given an ordered pair $(x_{i},f(x_{i}))$ as a share. This scheme still reliable and secure even when misfortunes destroy half the pieces and security breaches expose all but one of the remaining pieces. This scheme is \emph{perfect}, since all qualified subsets can reconstruct the secret and unqualified subsets cannot determine any information about the secret. The scheme is called \emph{ideal}, if $x_{i}$ is publicly revealed so that the share of participant $p_{i}$ becomes just $f(x_{i})$ and then the size of each share equals the size of the secret. The scheme proposed by Blakley is based on geometries over finite fields, it's perfect and can be modified slightly to become ideal, as explained by Ernest \cite{07}.\vspace{0.2cm} \\
Ito et al. have generalized the concept of threshold scheme and showed that, given any \emph{monotone access structure} $\Gamma$, i.e., for $Y\in\Gamma$, if $Y\subset Y'$ then $Y'\in\Gamma$,  there exist a perfect secret sharing scheme to realize the structure \cite{10,09}. Benaloh and Leichter  proposed a different algorithm that has a lower \emph{information rate} than Ito's et al. construction \cite{11}. In both constructions, the information rate decreases exponentially as a function of the number of participants $n=|P|$. The information rate, noted $\rho$, is considered as a measure of the efficiency of a secret sharing scheme. It is defined as the ratio between the secret size and the maximum size of the shares $S$, that is, $\rho = \frac{\log_2(\left|K \right|)}{\log_2(\left|S\right|)}$ \cite{07}. Other measures can also be considered such as \emph{the average information rate}, which is defined as the ratio between the length of the secret and the arithmetic mean of the length of all shares and expressed as follow $\widetilde{\rho}=\frac{n\log_2(\left|K\right|)}{\sum_{i=1}^{n}{\log_2(\left|S_i\right|)}}$ \cite{12}. \vspace{0.2cm} \\
Another approach based on the \emph{multilevel access structures} was presented by Simmons in 1988. In this approach each participant is assigned a level which is a positive integer and the access structure consists of those subsets which contain at least $r$ participants all of level at most $r$. That means for instance if $r=3$, then $3$ participants of level 3 can determine the secret, and also $1$ participant of level $1$ and one other participant of level $2$ and one participant of level $3$ can determine the secret, for more details see for instance \cite{08}. In \cite{07} Brickell shown that given any multilevel access structure, there exists $q_{0}$ such that for any $q$ a prime power with $q > q_{0}$, there is an ideal secret sharing scheme realizing this access structure over $GF(q)$.  \vspace{0.2cm} \\
There are also another approaches based on \emph{graph access structure} that have received a considerable attention. In the most of these approaches, many researchers have proposed different constructions of a perfect secret sharing scheme based on uniform access structures which contains qualified subsets all of the same cardinality $m$. In these constructions, participants are represented by the vertices of a graph $G$, the uniform access structure $\Gamma$ is based on the concept of adjacent vertices and represented by the edges, for more details see for instance \cite{04,18,03,14,06,13}. In \cite{02} a novel approach to design a graph access structure, which is based on the concept of non-adjacent vertices, was proposed. In this approach, an independent dominating set of vertices in a graph $G$ was introduced and applied as a novel idea to construct a perfect secret sharing scheme such that the vertices of the graph represent the participants and the dominating set of vertices in $G$ represents the minimal qualified set.
\section{The proposed construction algorithm}
\label{sec:2}
Shamir \cite{01} had specified that one of the useful properties of the proposed threshold scheme is that by using tuples of polynomial values as parts, it is possible to get a hierarchical scheme in which the number of parts needed to determine the secret depends on the importance of the participants. He also brought a brief explanation based on an example of a company's check signature. The motivation of this paper is to propose a novel construction algorithm of an ideal secret sharing scheme which is based on the hierarchical concept of companies and in which the access structure is not uniform.The proposed construction algorithm include two phases which are achieved by the dealer who can, for instance, be represented by the board of directors $(BOD)$ at a company.
\subsection{The initialization phase}
\label{sec:2.1}
The hierarchical concept of any company is illustrated through its organization chart which is represented by a tree $T=\left(V,E\right)$ such that:
\begin{itemize}[leftmargin=0.75cm,label=\textbullet,itemsep=0.2cm]
    \item The height of $T$ corresponds to the number of hierarchical levels at the company, denoted $h$, and each hierarchical level is denoted $N_{j}$, for $j=1,\ldots,h$.
		\item The set of vertices $V$ corresponding to the company's employees represents the set of participants $P$. As each participant $i$ belong to a specified level $j$, we denote by $P_{ij}$ such participant.
  \item The set of edges $E$ corresponds to the hierarchical relations between participants (employees).
\end{itemize}
Figure~\ref{figure1} given bellow, illustrates an organization chart of a company with $9$ employees and $3$ hierarchical levels.
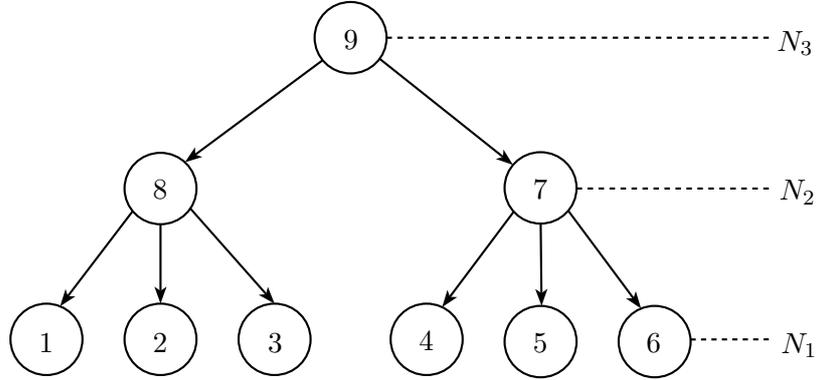
\begin{figure}[H]
\begin{center}
\psset{xunit=1.0cm,yunit=1.0cm,algebraic=true,dimen=middle,dotstyle=o,dotsize=5pt 0,linewidth=0.8pt,arrowsize=3pt 2,arrowinset=0.25}
\begin{pspicture*}(1.1343222434871767,5.145150786292249)(12.278464004306558,10.870433456467174)
\psaxes[labelFontSize=\scriptstyle,xAxis=true,yAxis=true,Dx=0.5,Dy=0.5,ticksize=-2pt 0,subticks=2]{->}(0,0)(1.1343222434871767,5.145150786292249)(12.278464004306558,10.870433456467174)
\rput[tl](8.4,8.11){$7$}
\rput[tl](3.4,8.11){$8$}
\rput[tl](5.91,10.1){$9$}
\pscircle(6.,10.){0.47332386778193497}
\pscircle(3.5,8.){0.4733238677819355}
\pscircle(8.498584457137188,8.008432758609183){0.4733238677819355}
\pscircle(7.,6.){0.4733238677819355}
\pscircle(8.498584457137188,5.970241233049502){0.4733238677819355}
\pscircle(9.998584457137191,5.970241233049502){0.4733238677819355}
\psline{->}(5.63253493006577,9.701663099516198)(3.821167704047247,8.347687776157102)
\psline{->}(6.382101423977742,9.720654361755491)(8.136427546873868,8.313191443736576)
\pscircle(2.,6.){0.4733238677819355}
\pscircle(3.5,6.){0.47332386778193597}
\pscircle(5.,6.){0.4733238677819355}
\psline{->}(3.133963422697201,7.699912166377071)(2.1790295638113095,6.438159673057189)
\psline{->}(3.5,7.526676132218064)(3.5,6.473323867781936)
\psline{->}(3.8924124648720424,7.7353342840013095)(5.,6.473323867781936)
\psline{->}(8.145615608361842,7.693078234927885)(7.203445611009179,6.427370292806079)
\psline{->}(8.499902247178753,7.535110725273044)(8.511457555565038,6.4433900121337535)
\psline{->}(8.862500764727278,7.705777170717524)(9.820000835131868,6.408582850698923)
\rput[tl](11.61,10.09){$N_{3}$}
\rput[tl](11.66,6.08){$N_{1}$}
\rput[tl](11.64,8.11){$N_{2}$}
\psline[linestyle=dashed,dash=2pt 2pt](6.473323867781935,10.)(11.5,10.)
\psline[linestyle=dashed,dash=2pt 2pt](8.971841401203262,8.000473568619492)(11.5,8.)
\psline[linestyle=dashed,dash=2pt 2pt](10.471076911116583,5.998283428836976)(11.5,6.)
\rput[tl](1.9,6.08){$1$}
\rput[tl](3.4,6.08){$2$}
\rput[tl](4.91,6.08){$3$}
\rput[tl](6.9,6.1){$4$}
\rput[tl](8.4,6.08){$5$}
\rput[tl](9.89,6.08){$6$}
\end{pspicture*}
\caption{Company organization chart $T$ with 9 employees.}
\label{figure1}
\end{center}
\end{figure}
\noindent In the initialization phase, the dealer proceed to the construction of the access structure $\Gamma$ containing all the qualified subsets. A subset $X$ of $P$ is considered as qualified if and only if:
\begin{enumerate}[leftmargin=0.75cm,label=($\roman*$),itemsep=0.2cm]
  \item $X$ contains more than one participant. No participant will have the veto right for reconstructing the secret alone, especially the first manager. This condition is formulated by: \\
	$$\sum \limits_{P_{ij} \in X}j \geq h+1.$$
  \item The elements of $X$ cannot all be at the same hierarchical level, in order to reduce the risk of corruption. This condition is expressed by: \\
	$$\left|X \cap N_{j} \right| \leq  \left\lceil \frac{h+1}{j}\right\rceil -1,\ \textrm{for} \ j=1,\ldots, h. $$
\end{enumerate}
The access structure $\Gamma$ is then:
$$\Gamma= \left\{ X \subset P \ : \ \sum \limits_{P_{ij} \in X}j \geq h+1 \ \textrm{and} \  \left|X \cap N_{j} \right| \leq  \left\lceil \frac{h+1}{j}\right\rceil -1,\ \textrm{for} \ j=1,\ldots, h\right\}.$$
The minimum access structure $\Gamma_{0}$ is then $\Gamma_{0} = \left\{ X \in\Gamma: \forall X' (X' \varsubsetneq X\Rightarrow X'\not\in \Gamma )\right\}$.
\subsection{The decomposition phase}
\label{sec:2.2}
In this phase, the dealer:
\begin{itemize}[leftmargin=0.75cm,label=\textbullet,itemsep=0.2cm]
\item Choose a prime power number $q$;
\item Select the secret to share $K=(k_{1},\ldots,k_{h})$ that he encodes in the finite field $GF(q)$;
\item Generate randomly one value $a_{0}$ in $GF(q)$;
\item Construct the polynomial $f(x)$ of degree $\emph{h}$:
$$f(x)= a_{0}+k_{1} x+\cdots+k_{h} x^h;$$
\item Calculate and distribute the shares to all participants. The share given to each participant $P_{ij}$, denoted $S_{ij}$, consists on two parts. The first one is publicly revealed and correspond to there login $i$ and hierarchical level $j$. The second part is sent in private and consists on $j$ values of ordered pairs $\left(x_{i1},f(x_{i1})\right), \ldots,\left(x_{ij},f(x_{ij})\right)$, so that the number of participants who can pool their shares to reconstruct the secret depends on their importance.
\end{itemize}
The following algorithm resumes the proposed construction of secret sharing scheme.
\begin{center}
\begin{algorithm}[H]
{\caption{Construction of secret sharing scheme}
\label {Alg1}
\begin{flushleft}
\algorithmicrequire
\begin{enumerate}
\item The set of company's participants $P=\left\{P_{ij},\ i=1,\ldots,n; \ j=1,\ldots,h\right\}$;\\
\item A prime power $q$;\\
\item The polynomial $f(x)= a_{0}+k_{1} x+\cdots+k_{h} x^h$. \\
\end {enumerate}
\algorithmicensure The set of shares assigned to participants $S= \left\{S_{ij},\ i=1,\ldots,n; \ j=1,\ldots,h\right\}$.
\rule{\linewidth}{0.5pt}\vspace{-0.1cm}
\end{flushleft}
\begin{algorithmic}[1]
\State For each participant $P_{ij}$, calculate $x_{im}=1+mih,\ m=1,\ldots,j$;
\State Calculate $S_{ij}=\left(i,j,\left(x_{i1},f(x_{i1})\right),\left(x_{i2},f(x_{i2})\right),\ldots,\left(x_{ij},f(x_{ij})\right)\right),\ i=1,\ldots,n;\ j=1,\ldots,h$;
\State \textbf{Return:} $S= \left\{S_{ij},\ i=1,\ldots,n;\ j=1,\ldots,h\right\}.$
\end{algorithmic}
}
\end{algorithm}
\end{center}
\noindent According to Horner's method, Algorithm~\ref{Alg1} can be achieved, in the worst case, in $\textit{O}(nH)$ time complexity.
\section{The proposed reconstruction algorithm}
\label{sec:3}
Let $K=(k_{1},\ldots,k_{h})$ be the secret shared over the finite set of participants $P$ by application of Algorithm~\ref{Alg1}. According to the polynomial chosen by the dealer for calculating and distributing the shares, a group of participants $X$ who want to collaborate with their shares in order to recover the secret $K$, should in first reconstruct the polynomial $f$, which can be done by interpolation, for that $X$ should own at least $h+1$ values of ordered pairs, $(x_{1},f(x_{1})),\ldots,(x_{h+1},f(x_{h+1}))$. The secret $K$ can be recover by applying the logical XOR operator on the $k_{i}$'s deduced from $f$:
$$K=k_{1}\oplus k_{2}\oplus \cdots \oplus k_{h}.$$
The proposed reconstruction is summarized in Algorithm~\ref{Alg2}.
\begin{theorem}\label{th1}
The constructed secret sharing scheme is perfect.
\end{theorem}
\begin{proof} Let $X$ be a qualified subset of participants, then the conditions (\emph{i}) and (\emph{ii}), in the initialization phase~\ref{sec:2.1} above, are satisfied.
According to the decomposition phase~\ref{sec:2.2}, each $P_{ij}$ belonging to $X$ owns as much values of $(x,f(x))$ as his level $j$, $(x_{i1},f(x_{i1})),\ldots,(x_{ij},f(x_{ij}))$. Thus, $X$ owns at least $h+1$ values of $(x,f(x))$ and can recover $f(x)$, by using interpolation, and then the secret $K$ by applying the logical XOR operator on the $k_{i}$'s deduced from $f$. Therefore, any qualified subset can reconstruct the secret.\vspace{0.2cm} \\
\noindent Now, let $X$ be an unqualified subset of participants, then one of the conditions (\emph{i}) and (\emph{ii}), in the initialization phase~\ref{sec:2.1}, is not satisfied. If the condition (\emph{i}) is not, $X$ owns less than $h+1$ values of $(x,f(x))$, which don't allow the reconstruction of $f(x)$. In the other hand, as the elements of $X$ cannot all be at the same hierarchical level, if the condition $(ii)$ is not satisfied, the system denies access. Therefore, any unqualified subset has no information about the secret.
\end{proof}
\begin{algorithm}[H]
{\caption{ Reconstruction of a secret $K$}
\label {Alg2}
 \begin{flushleft}
\algorithmicrequire
\begin{enumerate}
\item A subset of participants $X\subset P$;
\item The set of hierarchical levels, $N_{1},\ldots,N_{h}$;
\item The shares of participants belonging to $X$.
\end {enumerate}
\algorithmicensure
\begin{enumerate}
  \item The secret $K=k_{1}\oplus k_{2}\oplus \cdots \oplus k_{h}$ or
 \item The system denies access.
\end {enumerate}
\rule{\linewidth}{0.5pt}\vspace{-0.1cm}
\end{flushleft}
\begin{algorithmic}[1]
\Statex \textbf{If} $X\in\Gamma$ \textbf{Then} Apply interpolation to reconstruct $f(x)$ and then the secret $K$
\Statex \hspace{1.54cm} \textbf{Else} The system denies access and displays "The subset is not qualified".
\end{algorithmic}
}
\end{algorithm}
\section{The efficiency of the proposed secret sharing scheme}\label{sec:4}
To measure the efficiency of the proposed secret sharing scheme, we consider the information rate \mbox{$\rho = \frac{\log_2(\left|K \right|)}{\log_2(\left|S\right|)}$}, where $S$ is the maximum share.
\begin{theorem}\label{th2}
The constructed secret sharing scheme is ideal.
\end{theorem}
\begin{proof} The secret $K=(k_{1},\ldots,k_{h})$ is an $h$-dimensional vector such that each $k_{i},\ i=1\ldots,h$, is in $GF(q)$. The $k_{i}$'s length is then equal to $\log_2(q)$. According to the decomposition phase~\ref{sec:2.2}, each share $S_{ij}$ is represented by a vector of $j+2$ components, in which $j$ components are private. The maximum share $S$ is the one corresponding to the first manager of the company which is at the high level $h$, its length is then equal to $h\log_2(q)$. Hence, $\rho=1$.
\end{proof}
\section{Security analysis}\label{sec:5}
The two main security requirements in a secret sharing scheme are confidentiality and authentication. Confidentiality is about ensuring that the information is only available to the qualified subsets, while the authentication is intended to ensure that each participant trying to collaborate in order to reconstruct the secret, is the one he claims to be.\vspace{0.2cm} \\
In this paper, confidentiality has been demonstrated in Theorem~\ref{th1} by proving that the proposed secret sharing scheme is perfect, while authentication is ensured by denying the access of all types of attacks. In fact, in such protocols, two types of attacks can arise: the insider and outsider attacks.\vspace{0.2cm} \\
For the outsider attacks, where the attackers are not belonging to the system, the attacker aims to recover the secret by trying all possible combinations. As the secret $K$ is an $h$-dimensional vector in which each component is in $GF(q)$, the number of possible combinations increases according to the number of hierarchical levels $h$. Thus, the brute force attack becomes a combinatorial explosion. \vspace{0.2cm} \\
For the insider attacks, where the attackers are belonging to the system but consist on an unqualified subset of participants, as all parameters are public in the proposed scheme except the secret $K$, three types of insider attacks can arise:
\begin{itemize}[leftmargin=0.75cm,label=\textbullet,itemsep=0.2cm]
  \item The first case consists on participant in level $N_{i}$ who may pretend to be a participant of another lower level $N_{j},\ j < i$, and use only a part of his share, in order to escape the condition $(\emph{ii})$ described in the initialization phase~\ref{sec:2.1}. The following conditions $(\emph{iii})$ and $(\emph{iv})$ are then included in the proposed scheme and checked before proceeding to the reconstruction algorithm~\ref{Alg2}. In the case where these conditions are not satisfied, the system generates an authentication error and display an attack attempt message without executing the reconstruction algorithm~\ref{Alg2}.\vspace{0.2cm} \\
For each given share $$S_{ij}=\left(i,j,\left(x_{i1},f(x_{i1})\right),\left(x_{i2},f(x_{i2})\right),\ldots,\left(x_{ij},f(x_{ij})\right)\right),\ i=1,\ldots,n;\ j=1,\ldots,h,$$
\begin{enumerate}[leftmargin=0.75cm,label=($\roman*$),itemsep=0.5em]
\setcounter{enumi}{2}
\item The login $i$ corresponds to a participant of the level $j$. This condition is formulated by:
$$\forall S_{ij}, \ i=1,\ldots,n \ \textrm{and} \ j=1,\ldots,h; \ P_{ij} \in N_{j}.$$

\item Each ordered pairs $(x_{im}, f(x_{im})),\ m=1,\ldots,j$, corresponds to the one sent by the dealer to the participant $i$ belonging to the level $j$. This condition is expressed by:
$$\forall S_{ij},\ i=1,\ldots,n \ \textrm{and} \ j=1,\ldots,h;\ \forall x_{im}, \ m=1,\ldots,j; \ x_{im}=1\ (\bmod\ ih) \ \textrm{and} \ \left\lfloor\frac{x_{im}}{ih}\right\rfloor\leq j,$$
where $\lfloor.\rfloor$ denotes the floor function.
\end{enumerate}
  \item The second case of insider attacks consists on participants in the same level $N_{i}$, who are not allow to collaborate with their shares, according to condition ($\emph{ii}$), in Section \ref{sec:2.1}, trying to merge their shares to have only one and pretend to be a participant of another higher level $N_{j},\ j > i$. This case is treated as the first case described above.
  \item The last case of insider attacks consists on participant in level $N_{i}$, who may pretend to be a participant of another higher level $N_{j},\ j > i$, and try to calculate another value of $f(x)$. This case is similar to the outsider attacks described above.
\end{itemize}
\section{Didactic example}
\noindent Let consider the case of a company whose organization chart is represented by the tree $T$ given \mbox{in Figure ~\ref{figure1}} above. According to the initialization phase ~\ref{sec:2.1}:
	 \begin {itemize} [leftmargin=0.75cm,label=\textbullet,itemsep=0.5em]
        \item The number of hierarchical levels $h=3$.
		\item The set of participants $ P=\left\{P_ {11},P_{21},P_{31},P_{41},P_{51},P_{61},P_{72},P_{82},P_{93}\right\}$.
		\item  According to their hierarchical levels, participants are assigned as follow: \vspace{0.2cm} \\
 $N_1= \{P_{11},P_{21},P_{31},P_{41},P_{51},P_{61}\}$; \vspace{0.2cm} \\
 $N_2= \{P_{72},P_{82}\}$;\vspace{0.2cm} \\
 $N_3= \{P_{93}\}$.
\item The access structure $\Gamma_0$ containing all the minimal qualified subsets is given as follow:
\begin{equation}
\begin{split}
\notag \Gamma_0=&\{\{P_{93},P_{11}\},\ \ \{P_{93},P_{21}\},\ \ \ \{P_{93},P_{31}\},\ \ \{P_{93},P_{41}\},\ \ \{P_{93},P_{51}\},\ \ \{P_{93},P_{61}\},\ \ \ \{P_{93},P_{72}\},\\
&\{P_{93},P_{82}\},\{P_{82},P_{11},P_{21}\},\{P_{82},P_{11},P_{31}\},\{P_{82},P_{11},P_{41}\},\{P_{82},P_{11},P_{51}\},\{P_{82},P_{11},P_{61}\}, \\
&\{P_{82},\ P_{21},\ P_{31}\},\ \ \{P_{82},\ P_{21},\ P_{41}\},\ \ \{P_{82},\ P_{21},\ P_{51}\},\ \ \{P_{82},\ P_{21},\ P_{61}\},\ \{P_{82},\ P_{31},\ P_{41}\},\\
&\{P_{82},\ P_{31},\ P_{51}\},\ \ \{P_{82},\ P_{31},\ P_{61}\},\ \ \{P_{82},\ P_{41},\ P_{51}\},\ \ \{P_{82},\ P_{41},\ P_{61}\},\ \{P_{82},\ P_{51},\ P_{61}\},\\
&\{P_{72},\ P_{11},\ P_{21}\},\ \ \{P_{72},\ P_{11},\ P_{31}\},\ \ \{P_{72},\ P_{11},\ P_{41}\},\ \ \{P_{72},\ P_{11},\ P_{51}\},\ \{P_{72},\ P_{11},\ P_{61}\},\\
&\{P_{72},\ P_{21},\ P_{31}\},\ \ \{P_{72},\ P_{21},\ P_{41}\},\ \ \{P_{72},\ P_{21},\ P_{51}\},\ \ \{P_{72},\ P_{21},\ P_{61}\},\ \{P_{72},\ P_{31},\ P_{41}\},\\
&\{P_{72},\ P_{31},\ P_{51}\},\ \ \{P_{72},\ P_{31},\ P_{61}\},\ \ \{P_{72},\ P_{41},\ P_{51}\},\ \ \{P_{72},\ P_{41},\ P_{61}\},\ \{P_{72},\ P_{51},\ P_{61}\}\}
\end{split}
\end{equation}
\end{itemize}
\noindent Suppose for instance that the key $K$ is 32-bit integer and $q=4294967311$ a prime number greater than $2^{32}-1$. Based on the decomposition phase~\ref{sec:2.2}, let consider $k_{1}=4967295$, $k_{2}=94967$, $k_{3}=9496729$ and $a_{0}=429496$. The polynomial chosen by the dealer is then $$f(x)=429496+4967295x+94967x^{2}+9496729x^{3},$$ and the shares given to participants are:\vspace{0.2cm} \\
$S_{93}=(9,3,(x_{91},f(x_{91})),(x_{92},f(x_{92})),(x_{93},f(x_{93})))=(9,3,(28,2527731964),(55,31222823),(82,1673628957))$;\vspace{0.2cm} \\
$S_{72}=(7,2,(x_{71},f(x_{71})),(x_{72},f(x_{72})))=(7,2,(22,2492596253),(43,3826770342))$;\vspace{0.2cm} \\
$S_{82}=(8,2,(x_{81},f(x_{81})),(x_{82},f(x_{82}))=(8,2,(25,2541468297),(49,1061011979))$; \vspace{0.2cm} \\
$S_{11}=(1,1,(x_{11},f(x_{11})))=(1,1,(4,629608804));\vspace{0.2cm} \\
S_{21}=(2,1,(x_{21},f(x_{21})))=(2,1,(7,3297231991));\vspace{0.2cm} \\
S_{31}=(3,1,(x_{31},f(x_{31})))=(3,1,(10,966393524));\vspace{0.2cm} \\
S_{41}=(4,1,(x_{41},f(x_{41})))=(4,1,(13,3765498123));\vspace{0.2cm} \\
S_{51}=(5,1,(x_{51},f(x_{51})))=(5,1,(16,348113953))$;\vspace{0.2cm} \\
$S_{61}=(6,1,(x_{61},f(x_{61})))=(6,1,(19,842645734))$.\vspace{0.2cm} \\
It's clear that each qualified subset belonging to $\Gamma_0$ can recover the secret $K$. \vspace{0.2cm} \\
Let's take for instance the qualified subset $X=\{P_{82},P_{11},P_{21}\}$. According to the reconstruction Algorithm ~\ref{sec:3}, the polynomial $f$ can be reconstruct by applying interpolation. \vspace{0.2cm} \\
The polynomial $L$ defined bellow is the unique polynomial of degree at most $h$ satisfying \mbox{$L(x_{i})=y_{i}=f(x_{i})$}:
$$ L(x)=\sum \limits_{j=0}^{h} f(x_{j}) l_{j}(x),\ \textrm{where}\ l_{j}(x)=\prod\limits_{\underset{i\neq j}{i=0}}^{h} \left( \frac{x-x_{i}}{x_{j}-x_{i}}\right).$$
For the considered qualified subset $X$, the $h$ known values of $(x,f(x))$ are:
\begin{table}
\begin{center}
\begin{tabular}{|l|l|}
\hline
$x_{0}=x_{81}=25$ & $f(x_{0})=2541468297$ \\ \hline
$x_{1}=x_{82}=49$ & $f(x_{1})=1061011979$ \\ \hline
$x_{2}=x_{11}=4$ & $f(x_{2})=629608804$ \\ \hline
$x_{3}=x_{21}=7$ & $f(x_{3})=3297231991$ \\ \hline
\end{tabular}
\end{center}\vspace{0.2cm}
\caption{$(x,f(x))$ values of qualified subset.}
\end{table}
Lagrange polynomials are calculated as follow:
\begin{center}
\begin{tabular}{l}
$l_{0}(x) = \dfrac{(x-49)(x-4)(x-7)}{(25-49)(25-4)(25-7)}=\dfrac{1}{9072}\left(-x^{3}+60x^{2}-567x+1372\right)$,\vspace{0.2cm}\\
$l_{1}(x) = \dfrac{(x-25)(x-4)(x-7)}{(49-25)(49-4)(49-7)}=\dfrac{1}{45360}\left(x^{3}-36x^{2}+303x-700\right)$,\vspace{0.2cm}\\
$l_{2}(x) = \dfrac{(x-25)(x-49)(x-7)}{(4-25)(4-49)(4-7)}=\dfrac{1}{2835}\left(-x^{3}+81x^{2}-1743x+8575\right)$,\vspace{0.2cm}\\
$l_{3}(x) = \dfrac{(x-25)(x-49)(x-4)}{(7-25)(7-49)(7-4)}=\dfrac{1}{2268}\left(x^{3}-78x^{2}+1521x-4900\right).$
\end{tabular}
\end{center}
Hence
\begin{center}
\begin{tabular}{lll}
$L(x)$ &=&$2541468297\ l_{0}(x)+1061011979\ l_{1}(x)+629608804\ l_{2}(x)+3297231991\ l_{3}(x)\ (\bmod\ q)$\vspace{0.2cm}\\
&=& $f(x)$.
\end{tabular}
\end{center}
Therefore
\begin{table}[H]
\begin{center}
\begin{tabular}{|c|c|c|}
\hline
& Decimal value & Binary value \\ \hline
$k_{1}$ & $4967295$ & $010010111100101101111111$ \\ \hline
$k_{2}$ & $94967$ & $000000010111001011110111$ \\ \hline
$k_{3}$ & $9496729$ & $100100001110100010011001$ \\ \hline
$K=k_{1}\oplus k_{2}\oplus k_{3}$ & $14307601$ & $110110100101000100010001$
\\ \hline
\end{tabular}
\end{center}
\caption{Reconstruction of the secret $K$.}
\end{table}
\vspace{0.4cm}

\noindent \textbf{In case of insider attacks}: as a first case of an insider attack, let's take the case in which the subset $\left\{ P_{82}, P_{72}\right\}$, who is not qualified, try to reconstruct the secret by using the $P_{82}$ share's as if it concerned those corresponding to participants $P_{11}$ and $P_{21}$. For instance, instead of introducing the share $S_{82}$ given above, $P_{82}$ will introduce the following vectors $S_{11}^{'}$ and $S_{21}^{'}$ as shares of $P_{11}$ and $P_{21}$, respectively: \\
$$S_{11}^{'}=(1,1,(x_{81},f(x_{81})))=(1,1,(25,2541468297)),$$
$$S_{21}^{'}=(1,1,(x_{82},f(x_{82})))=(2,1,(49,1061011979)).$$
The condition ($\emph{iv}$), in Section \ref{sec:5}, is not satisfied in this case, since: \vspace{0.2cm}
 $$x_{81}=1\ (\bmod\ 3),\ \textrm{but}\ \left\lfloor\frac{x_{81}}{3}\right\rfloor > 1,$$
 $$x_{82}=1\ (\bmod\ 6),\ \textrm{but}\ \left\lfloor\frac{x_{82}}{6}\right\rfloor > 1.$$
\vspace{0.2cm}\\The system generates then an authentication error and display an attack attempt message. \vspace{0.2cm}\\
As a second case of an insider attack, let's take the case in which the subset $\left\{ P_{11}, P_{21}, P_{31}, P_{41}\right\}$, who is not qualified, according to condition ($\emph{ii}$), in Section \ref{sec:2.1}, try to reconstruct the secret by merging the shares of $P_{31}$ and $P_{41}$ and pretending to be the subset $\left\{ P_{11}, P_{21}, P_{72}\right\}$ for instance. \vspace{0.2cm}\\
In this case, instead of introducing the shares $S_{31}$ and $S_{41}$ given above, a merged share $S_{72}^{'}$ is introduced as if it was the one corresponding to the participant $P_{72}$:
$$S_{72}^{'}=(7,2,(x_{31},f(x_{31})),(x_{41},f(x_{41})))=(7,2,(10,966393524),(13,3765498123)).$$
The condition ($\emph{iv}$), in Section \ref{sec:5}, is not satisfied in this case, since
 $$\left\lfloor\frac{x_{31}}{21}\right\rfloor <2,\ \textrm{but} \  x_{31}\neq1\ (\bmod\ 21),$$
 $$\left\lfloor\frac{x_{41}}{21}\right\rfloor <2,\ \textrm{but} \ x_{41}\neq1\ (\bmod\ 21).$$
The system generates then an authentication error and display an attack attempt message. \vspace{0.2cm}\\
\noindent \textbf{In case of outsider attacks}: as all coefficients of $f$ are taken in $GF(q)$, the attackers should try $q^{h+1}$ \mbox{possible} combinations to reconstruct $f$. In our example, this requires $4294967311^{4}$ possibilities, that exceeds $2^{116}$.
\section{Conclusion}

\noindent In this paper, we first propose a novel construction of a secret sharing scheme, which is based on the hierarchical concept of companies. In the proposed scheme, polynomials are used over $GF(q)$ and the considered access structure is not uniform, since the number of parts needed to reconstruct the secret depends on the importance of the participants. We also present a reconstruction algorithm, in which the interpolation and the logical XOR are used to reconstruct the polynomial and recover the secret $K$, respectively. We show that the proposed scheme is perfect and ideal. Furthermore, the security of the proposed scheme is analyzed by discussing all possible kinds of attacks (insider and outsider) and proofing that confidentiality and authentication are ensured. Finally, we conclude by a detailed didactic example.


%
%



\end{document}